\documentclass{article}
\usepackage{fullpage}
\usepackage{booktabs}
\usepackage{multirow}
\usepackage[utf8]{inputenc}
\usepackage{amsthm,amssymb}
\usepackage{amsmath}
\usepackage{tcolorbox}
\usepackage{thmtools}
\usepackage{mathtools}
\usepackage{subcaption,thm-restate}
\usepackage[mathcal,mathscr]{eucal}
\usepackage{epsfig,graphicx,graphics,color}
\usepackage{caption}
\usepackage{enumerate}
\usepackage{enumitem}
\usepackage[sort]{cite}
\usepackage{titling}
\usepackage{hyperref}
\usepackage{verbatim}

\hypersetup{
    colorlinks=true,       
    linkcolor=blue,        
    citecolor=red,         
    filecolor=magenta,     
    urlcolor=cyan,         
    linktocpage=true
}
\usepackage{thm-restate}
\usepackage{cleveref}
\usepackage{float}
\usepackage{algpseudocode}
\usepackage{rotating}
\usepackage{makecell}
\usepackage{pdfpages}
\usepackage{graphicx}
\usepackage{subcaption}
\usepackage{algorithm}

\usepackage{longtable}
\usepackage{booktabs,siunitx,array,threeparttable}
\usepackage[acronym]{glossaries}

\algnewcommand{\Input}[1]{\State \textbf{Input:} #1}
\algnewcommand{\Output}[1]{\State \textbf{Output:} #1}
\algnewcommand{\Comment}[1]{\Statex \(\triangleright\) #1}

\newacronym{sg}{SG}{Storage Grid}
\newacronym{ul}{UL}{Unit Load}
\newacronym{amr}{AMR}{Automated Mobile Robot}

\theoremstyle{plain}
\newtheorem{theorem}{Theorem}
\newtheorem{lemma}[theorem]{Lemma}

\theoremstyle{definition}

\def\final{0}  
\def\iflong{\iffalse}
\ifnum\final=0  
\newcommand{\mnote}[1]{{\color{orange}[{\tiny \textbf{Malte:} \bf #1}]\marginpar{\color{orange}*}}}
\newcommand{\jnote}[1]{{\color{red}[{\tiny \textbf{Julian:} \bf #1}]\marginpar{\color{red}*}}}
\newcommand{\ynote}[1]{{\color{blue}[{\tiny \textbf{Yagmur:} \bf #1}]\marginpar{\color{blue}*}}}
\newcommand{\snote}[1]{{\color{purple}[{\tiny \textbf{Simone:} \bf #1}]\marginpar{\color{purple}*}}}
\else 
\newcommand{\mnote}[1]{}
\newcommand{\jnote}[1]{}
\newcommand{\ynote}[1]{}
\newcommand{\snote}[1]{}
\fi
\usepackage{authblk}

\title{Order Retrieval in Compact Storage Systems}

\date{}

\author{Malte Fliedner} \author{Julian Golak}
\affil{{\footnotesize Institute of Operations Management, University of Hamburg Business School, Hamburg, Germany. \\ \texttt{malte.fliedner@uni-hamburg.de; julian.golak@uni-hamburg.de}.}}
\author{Ya\u{g}mur Gül\footnote{Corresponding author} }\author{Simone Neumann}
\affil{{\footnotesize Institute of Sustainable Logistics and Mobility, University of Hamburg Business School, Hamburg, Germany. \\ \texttt{yagmur.guel@uni-hamburg.de; simone.neumann@uni-hamburg.de}.} }

\begin{document}
\maketitle
\begin{abstract}
\medskip
Growing demand for sustainable logistics and higher space utilization, driven by e-commerce and urbanization, increases the need for storage systems that are both energy- and space-efficient.
Compact storage systems aim to maximize space utilization in limited storage areas and are therefore particularly suited in densely-populated urban areas where space is scarce.
In this paper, we examine a recently introduced compact storage system in which uniformly shaped bins are stacked directly on top of each other, eliminating the need for aisles used to handle materials. 
Target bins are retrieved in a fully automated process by first lifting all other bins that block access and then accessing the target bin from the side of the system by a dedicated robot. 
Consequently, retrieving a bin can require substantial lifting effort, and thus energy.
However, this energy can be reduced through smart retrieval strategies.
From an operational perspective, we investigate how retrievals can be optimized with respect to energy consumption.

We model the retrieval problem within a mathematical framework. We show that the problem is strongly NP-complete and derive structural insights.
Building on these insights, we propose two exact methods: a mixed-integer programming (MIP) formulation and a dynamic programming algorithm, along with a simple, practitioner-oriented greedy algorithm that yields near-instant solutions.
Numerical experiments reveal that dynamic programming consistently outperforms state-of-the-art MIP solvers in small to medium sized instances, while the greedy algorithm delivers satisfactory performance, especially when exact methods become computationally impractical.
\medskip

\noindent \textbf{Keywords:} Order Picking, Warehouse Management, Algorithms,  Computational Complexity

\end{abstract}

\section{Introduction}


The trend in recent years has been towards an enormous number of (online) orders while customers tend to expect ever shorter delivery times and often same-day delivery.
At the same time, warehouse providers are under increasing pressure to operate more sustainably.
As a result, storage systems are also undergoing development and continue to evolve.
According to Statista, the global warehouse automation market is expected to grow from USD 23 billion in 2023 to USD 41 billion in 2027 \cite{Statista.2023}.
For some years now, warehouses have been using goods-to-picker systems, in which small robots move entire shelves to picking stations at high speed, thus enabling fast order picking.
However, in the case of centrally located warehouses where space is at a premium, compact storage systems that use autonomous robots for storage and retrieval are beneficial.
AutoStore\footnote{https://de.autostoresystem.com} is an example of a compact automated storage and retrieval system. Rather than having shelves with permanent aisles, individual bins (also known as boxes or stock keeping units) are stacked directly above and next to each other. These uniform bins contain the stored items and are referred to as unit loads (ULs). Storage and retrieval take place by autonomous robots via temporary aisles, which are created by lifting the \textit{blocking} ULs.
Worldwide, there are more than 1,650 systems in operation in the grocery, retail, and other sectors, and the number is growing very quickly. The system can be scaled in both length and width, with a maximum stacking height of 16.5 meters.
Very similar systems have been developed by Ocado, Jungheinrich (PowerCube), and Gebhardt, which have flexible bin heights. Related systems can be found at Attabotics and Exotec.

However, these systems typically only allow access iteratively from above or below. This means that ULs blocking access to a target UL must be retrieved from storage one by one. To overcome this limitation, a new compact storage concept was recently introduced\cite{fauve2022storage} that, unlike previous systems, allows access from the side. This creates flexible access aisles by temporarily lifting all blocking ULs out of the way.
In this paper, we investigate operational problems that arise in such compact storage systems.


\paragraph{Previous work.}

The forerunners of compact storage systems can be seen in deep lane storage systems, where products on shelves are stored one behind the other, see Stadtler (1996) \cite{stadtler1996operational}. Here, the degree of compactness is of course significantly lower than in cube-shaped systems. While there is a lot of research on these kinds of systems with permanent aisles, we do not regard these systems as compact storage systems and therefore do not consider them in the following. General reviews on warehousing and storage systems are provided by Azadeh et al. (2019) \cite{azadeh2019robotized}, Boysen et al. (2019) \cite{boysen2019warehousing}, Boysen and De Koster (2025) \cite{boysen202550}, and Winkelhaus and Grosse (2020) \cite{winkelhaus2020logistics}. An overview of the different types of problems that occur in parts-to-picker warehouses (where ULs are automatically retrieved and delivered to the pick stations) can be found in Boysen et al. (2023) \cite{boysen2023review}. The literature on compact storage systems itself is relatively limited and can be divided into papers on puzzle-based storage systems and a few papers on AutoStore. For a comprehensive overview, including these and different types of storage systems, see Fauvé et al. (2025) \cite{fauvé2024litrev}. 

Puzzle-based storage systems, also known as robotic live-cube compact storage systems, are cubic systems that have at least one empty location, referred to as an \textit{escort}. Each UL is stored on a shuttle and can move autonomously in two directions. A lift is used to move the ULs in the third direction (across different levels). The difference from the problem we are looking at is that the ULs in puzzle-based systems are autonomous via the linking and docking of their own shuttle with each individual UL and do not have to be transported by a few individual robots.
These systems can be found in automated parking systems, warehouses, and container handling.
They are analyzed mainly with respect to the retrieval times of the ULs, but also in general to find the optimal configuration for these kinds of systems.
This area is being researched in \cite{gue2007puzzle}, \cite{gue2013gridstore}, \cite{he2023reinforcement}, \cite{khai2024determining}, \cite{kota2015retrieval}, \cite{zaerpour2017small}, among others.

Only a few papers explicitly address compact storage systems like AutoStore, all of which are referred to as Robotic Compact Storage and Retrieval Systems (RCSRSs).
This class of systems is analyzed in \cite{cai2024robot}, \cite{chen2015bi}, \cite{chen2022simulation}, \cite{ko2022rollout}, \cite{tappia2017modeling}, \cite{yu2015class}, \cite{zou2018operating}. 
Warehouses with heterogeneous robots (a picker and a mover robot) are studied in \cite{kang2025warehouses}.
\cite{tutam2025paving} develops a new variant of a RCSRS, the Robotic Compact Vertical Carousel System, where blocking bins don't have to be replaced. 

The AutoStore system itself is analyzed by Trost et al. \cite{trost2022simulation}, who take a closer look at both the technical aspects and the throughput of the system. 
\cite{ha2024characteristics} compares AutoStore with other systems considering its space efficiency, robotic mechanisms, scalability, and energy efficiency, based on a preliminary overview. In another paper, the same author uses genetic algorithms to assign tasks to autonomous robots in an AutoStore system \cite{ha2024clustering}.
Energy efficiency in logistics, in general, including AutoStore, is studied in \cite{ariff2022energy}.
A case study on AutoStore's performance analysis was conducted by \cite{bremer2024systems}.
\cite{liu2024bins} study the bin relocation problem and compare it with the container relocation problem and the traveling salesman problem.

The system we are studying in this paper is a compact cube without aisles. Storage and retrieval take place from the side by creating temporary aisles. So far, this system has only been investigated by Fauvé et al. \cite{fauve2022storage}, where simple time and energy calculations were made. It can be seen as a starting point for research into this new type of storage system, which does not yet exist in the industry. In the paper at hand, we focus on analyzing the energy consumption during the retrieval process.

\paragraph{Our Contribution.}

Despite high demand and considerable interest from the industry in systems that are highly energy efficient and maximize space utilization, research into compact storage systems has been very limited to date.
With this paper, we aim to address this gap by investigating the following research question: What structural properties characterize energy-optimal retrieval batches in side-access compact storage, and how can these be exploited to design exact and scalable algorithms?
Our first contribution is to initiate the algorithmic study of compact storage systems with side-access. To this end, we formally introduce the problem and discuss different encoding schemes.
We then demonstrate that the problem is strongly NP-complete and show more structural insight. 
Building on these insight, we turn to algorithmic solutions.
We propose two algorithms to compute an optimal solution. 
For the first algorithm, we propose a mathematical programming formulation of the problem and employ a state-of-the-art solver to compute an optimal solution.
The second algorithm is based on dynamic programming, enhanced by dominance criteria to reduce the state space and improve efficiency. Finally, we propose a greedy algorithm algorithm based on simple rules, designed for straightforward implementation by practitioners.
We conclude the paper with experimental results demonstrating that the dynamic programming algorithm outperforms the state-of-the-art solver. 
Moreover, the greedy algorithm achieves satisfactory performance and scales effectively, whereas exact solution methods quickly encounter computational limitations in real-world scenarios.

\medskip
The remainder of the paper is structured as follows: Section~\ref{sec:problemDef} contains the problem definition. 
Section~\ref{sec:strucProperties}
establishes structural properties of the problem. Following this, we propose a MIP formulation in Section~\ref{sec:mathProg}.
Section~\ref{sec:algorithms} presents 
a dynamic programming approach and a greedy algorithm. The performance of these algorithms is analyzed in Section~\ref{sec:compStudy}. Finally, Section~\ref{sec:conclusions} provides a summary of our findings and proposes directions for future research.

\section{The Side-Access Compact Retrieval Problem}
\label{sec:problemDef}
In this section, we formally introduce the problem and present two alternative schemes to encode an input instance. 
In Section~\ref{sec:dense}, we give the problem definition using the encoding scheme that is efficient when the number of ULs that need to be retrieved is large relative to the size of the storage (= dense instance).
In contrast, in Section~\ref{sec:sparse}, 
we present an encoding scheme that is efficient when the number of ULs that need to be retrieved is small relative to the size of the storage (= sparse instance).

\subsection{Problem Definition and Dense Encoding}\label{sec:dense}
    We consider a \emph{storage cuboid}, which consists of multiple \emph{stacks}.
    A \emph{stack} is a vertical arrangement of ULs placed directly on top of each other without gaps.
    Note that stack heights need not be uniform.
    Furthermore, we consider a \emph{shuttle}, which is an autonomous vehicle that retrieves ULs from within the cuboid and delivers them to an external drop-off location.
	The cuboid is equipped with two distinct types of lifting mechanisms.
    On the sides of the cuboid, there are \emph{side-lifts} attached, which can raise shuttles to a vertical level.
	Each stack is equipped with a \emph{top-lift} mechanism positioned at the top of the storage cuboid. 
    The top-lift raises sub-stacks vertically to create temporary clearance, enabling the shuttle to maneuver within the cuboid.
	We refer to a row of stacks as a \emph{slice} of the cuboid.
	The shuttle moves inside the cuboid on a rail system that is attached to the top of the ULs.
    This rail system keeps the shuttle within its designated area (its slice) when it is inside the cuboid, restricts its movement to positions directly above the ULs, and prevents any rotation of the shuttle. 
	Throughout this paper, we assume that side-lifts are located on a single side of the cuboid. 
	Accordingly, all entry and exit points are located on this side, removing the need to make directional decisions for the shuttle. 
	This is a simplifying assumption that is found in practice, for example when the space is restricted from all other sides but one. 
	As a result, each slice can be operated independently, and the operational management of the entire cuboid can thus be decomposed into the operational management of its individual slices.
	Hence, we consider a single slice represented as a sequence of stacks \(T\coloneqq (1, \ldots, m)\), where \(m\coloneqq |T|\). 
    Each stack \(t\in T\) contains ULs stacked to height \(h(t)\). 
	Throughout this paper, stacks are indexed starting at \(1\) on the entry/exit side and increasing in the direction away from it, while heights are indexed from bottom to top.
    Throughout, we use \emph{left} to mean closer to the entry side and \emph{right} to mean farther from it, consistent with the orientation in our figures.
	We are given a customer order consisting of items located across multiple ULs within the slice. 
    To fulfill the order, the shuttle must retrieve a collection of ULs such that the customer order is covered.
	Although identical items may be stored in several ULs, we assume that the choice of which specific ULs to retrieve (target ULs) is made during a preceding planning phase.
    This is a widely used simplification in order retrieval problems, see, e.g., Schiffer et al. \cite{Schiffer2022}.
	Consequently, we are given a \emph{pick-list} of target ULs, denoted by \(\mathcal{B}\).
    For brevity, we will simply refer to them as \emph{targets}, and to the remaining ULs as \emph{non-targets}.
	We let \(n \coloneqq |\mathcal{B}|\).
	For each \(b\in \mathcal{B}\), we let \(s(b)\) denote its stack and \(h(b)\) its initial height, allowing for a natural overload of \(h\). 
    We often refer to the position of a UL by its stack/height pair \((t,h)\).

	The retrieval process is partitioned into a sequence of cycles \(C \coloneqq (1, 2,\ldots )\).
	In a cycle \(c\), the shuttle retrieves a sequence of targets \(Y_c \coloneqq (b_{c,1}, b_{c,2}, \dots, b_{c,|Y_c|})\), where each \(b_{c,i}\in\mathcal{B}\) represents the \(i\)-th target retrieved in cycle \(c\).
	Once a target is retrieved, it is delivered to the external drop-off location and is not returned to the system, and we do not allow the movement of non-targets.
	We let \(h_c(t)\) and \(h_c(b)\) denote the heights of stack \(t\in T\) and target \(b \in \mathcal{B}\) in the beginning of cycle \(c\).
	Specifically, we have that \(h_1(t) \coloneqq h(t)\) and \(h_1(b) \coloneqq h(b)\) for all \(t\) and \(b\).
	To retrieve a target, the shuttle is first elevated vertically by the side-lift to the target’s height and then all ULs that block the shuttle’s path on its horizontal trip are lifted by the top lifts, such that the target can be accessed.
	Thus, at the beginning of cycle \(c\), the operator decides on a \emph{clearance level} \(\ell_c(t) \in\{0,1, \ldots, h_c(t)\}\), for each stack \(t\in T\). 
	All ULs in \(t\) with height at least \(\ell_c(t)\) are lifted and remain lifted throughout the cycle. 
	Let \(e_c(t) \coloneqq h_c(t) - \ell_c(t)\) denote the number of ULs lifted from stack \(t\) in cycle \(c\).
    This temporarily lowers the heights of the stacks during the retrieval cycle.
    For all \(t \in T\), we let \(d_{c,i}(t)\) denote the \emph{residual height} after the first \(i\) elements \(b_{c,1},\dots,b_{c,i}\) of \(Y_c\) have been retrieved.
    Specifically, we have \(d_{c,0}(t) \coloneqq \ell_c(t)\). 
    After the \(i\)-th retrieval (\(i=1,\dots, |Y_c|\)), we update, for all \(t \in T\),
	\[
	d_{c,i}(t)\coloneqq
	\begin{cases}
		d_{c,i-1}(t)-1, &\text{if } t=s(b_{c,i}),\\[4pt]
		d_{c,i-1}(t),   &\text{otherwise}.
	\end{cases}
	\] 
	After the cycle is complete, the lifted ULs are lowered again. 
    If all targets are retrieved at the end of cycle \(c\), the retrieval process completes.
    Otherwise, we initiate cycle \(c+1\) and update heights with \(h_{c+1}(t) \leftarrow d_{c, |Y_c|}(t) + e_c(t)\) for all \(t \in T\).
    The height of a target that has not been retrieved may decrease when targets located in the same stack below it are retrieved during the cycle.
    That is, \(h_{c+1}(b) \leftarrow h_{c}(b) - |\{b'\in Y_c\mid s(b') = s(b), \ h(b') < h(b)\}|\) for all targets \(b\in \mathcal{B}\) that were not retrieved up to and including cycle \(c\). 
    
	A target can only be retrieved if it is \emph{accessible}.
    Suppose that in the \(i\)-th retrieval of cycle \(c\) we retrieve target \(b\in \mathcal{B}\) from its current position \((t,h)\). 
    We say that \(b\) is accessible if and only if the residual heights satisfy the following \emph{accessibility conditions}:
	\begin{enumerate}[label=(\alph*)]
		\item \label{cond:above} There is no UL directly above \(b\) to ensure that the shuttle can retrieve it, i.e. \(d_{c,i-1}(t) = h + 1\). 
		\item \label{cond:passage} All stacks \(s\) to the left of \(t\) must have a residual height of \(h\), so that there is a path for the shuttle, i.e. \(d_{c,i-1}(s) = h\) for \(s < t\).
	\end{enumerate}
	
	An input instance is defined by the tuple
	\(
	\mathcal{I} = (T, \mathcal{B}, (s(b), h(b))_{b\in\mathcal{B}}, (h(t))_{t\in\mathcal{T}}).
	\)
	A solution to an instance \(\mathcal{I}\) is given by a tuple
	\(
	\mathcal{X} = (C, (Y_c, \ell_c(t)_{t\in T})_{c \in C}),
	\) and is \emph{feasible} if and only if every UL is accessible at the moment it is retrieved, and all ULs in the pick-list \(\mathcal{B}\) are retrieved exactly once.
	We aim to minimize the system's total \emph{energy expenditure}. 
	The primary source of energy expenditure arises from lifting ULs. 
	The energy required to lift a UL corresponds to gravitational potential energy, which is computed as \(E = m g h\), where \(m\) denotes mass, \(g\) gravitational acceleration, and \(h\) lifting height. 
	Consequently, the energy required to lift a UL is directly proportional to its weight (mass times gravitational acceleration). 
    For simplicity, we assume that all ULs have uniform weight and are lifted by the same distance, so the total energy expenditure is directly proportional to the number of lifted ULs and their lifting heights.
	Without loss of generality, we further normalize weights so that lifting a single UL by one unit of height corresponds exactly to one unit of energy. 
	Thus, the \emph{energy expenditure} of a solution \(\mathcal{X}\) is defined as:
	\[
	E(\mathcal{X}) \coloneqq \sum_{c\in C} E(c),
	\]
    where \(E(c)\) is the energy expenditure of cycle \(c\), defined by \(E(c)\coloneqq \sum_{t\in T} e_c(t)\),
	where \(e_c(t)\) is the number of elevated ULs from stack \(t\) in cycle \(c\).
	A feasible solution that minimizes the total energy expenditure is referred to as \emph{optimal}. 
    We refer to the problem of finding an optimal solution as the \emph{Side-Access Compact Retrieval Problem} (SACRP).
	Observe that we encode the heights of the stacks in an array.
	We let \(H\) denote the maximum height of any stack in the slice, and the description length of an instance is then \(O((n\log m + m)\log H)\).

\paragraph{Example.}
	We illustrate the SACRP with the example shown in Figure~\ref{fig:example}.  
	Consider an instance, consisting of a layout with four stacks \(T\coloneqq \{1,2,3,4\}\) whose initial heights are
	\((h(1),h(2),h(3),h(4))=(5,4,2,4)\). 
	The pick-list contains six targets; \(\mathcal{B} \coloneqq \{b_{1}, b_{2}, b_{3}, b_{4},b_{5}, b_{6}\}\) whose initial positions are
	\(\bigl (s(b_i),h(b_i)\bigr)_{i = 1,\ldots,6 }
	\coloneqq ((1,4),\ (4,4),\ (1,3),\ (2,3),\ (4,3),\ (1,2)).
	\)
	We consider a solution that consists of two cycles. 
	In the first cycle, we lift the topmost ULs from stacks \(1\), \(2\), and \(4\).
	This allows us to retrieve targets \(b_1\), \(b_3\), \(b_4\), \(b_5\) and \(b_6\).
	In the second cycle, we only lift the topmost UL from stack \(2\) and can then retrieve the remaining target \(b_2\). In the first cycle, we lifted three ULs and in the second cycle we lifted one UL, resulting in an energy expenditure of \(4\).
    \begin{figure}
    \centering
    \begin{subfigure}[t]{.22\textwidth}
        \centering
        \includegraphics[page=2,width=\textwidth]{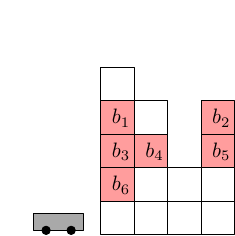}
        \caption{}
        \label{fig:example_sub1}
    \end{subfigure}
    \hfill
    \begin{subfigure}[t]{.22\textwidth}
        \centering
        \includegraphics[page=3,width=\textwidth]{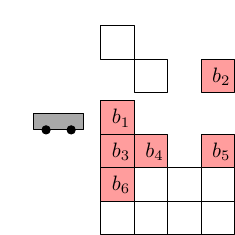}
        \caption{}
        \label{fig:example_sub2}
    \end{subfigure}
    \hfill
    \begin{subfigure}[t]{.22\textwidth}
        \centering
        \includegraphics[page=4,width=\textwidth]{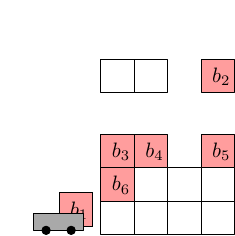}
        \caption{}
        \label{fig:sub3}
    \end{subfigure}
    \hfill
    \begin{subfigure}[t]{.22\textwidth}
        \centering
        \includegraphics[page=5,width=\textwidth]{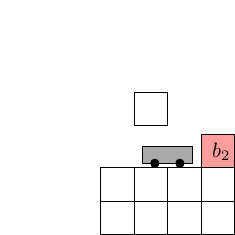}
        \caption{}
        \label{fig:example_sub4}
    \end{subfigure}
    \caption{Illustration of the entire retrieval process that is described in the example. Targets are highlighted in red. The total energy expenditure of the retrieval process is \(4\). 
    a) shows the initial configuration before any target is retrieved. b) depicts the configuration in the beginning of cycle 1, in which targets \(b_1\), \(b_3\), \(b_4\), \(b_5\) and \(b_6\) are retrieved in a single cycle with an energy expenditure of \(3\). In c) first target of cycle 1 (\(b_1\)) is retrieved. d) shows the configuration in the beginning of cycle 2, in which target \(b_2\) is retrieved with an energy expenditure of \(1\).
    }
    \label{fig:example}
\end{figure}

\subsection{Sparse Encoding}\label{sec:sparse}

    In the previous section, stack heights were encoded as an array. 
    While this is intuitive, it is wasteful for sparse instances, in which the pick list is small relative to the slice. 
    We therefore propose a \emph{sparse encoding} that is more efficient when \(n \ll m\).
    
    Consider an instance and a solution to the SACRP, and let \(c\) be a cycle with retrieval sequence \(Y_c\).  
    For every \(t \in T\), let \(h_r(t)\) denote the height of the topmost target in \(Y_c\) among stacks with \(s > t\) in the beginning of cycle \(c\) and some placeholder value if such target does not exist.  
    Let \(s\) be the rightmost stack containing a target \(b \in Y_c\), and let \(h\) denote the height of the topmost target in stack \(s\).  
    In an optimal solution, the clearance levels follow directly from the choice of \(Y_c\).  
    Specifically, the accessibility conditions imply that for every stack \(t < s\) we require \(\ell_c(t) = h_r(t)\); for stack \(s\) we require \(\ell_c(s) = h+1\); and for every stack \(t > s\) we require \(\ell_c(t) = h_c(t)\).  
    
    Furthermore, let \(X_c\subseteq \mathcal{B}\) denote the set of targets retrieved in the cycles preceding \(c\). 
    We say that target \(b\) is at level \(i\) if exactly \(i\) targets have already been retrieved (i.e., are in \(X_c\)), that are located in the same stack \(s(b)\) and at heights below \(h(b)\).
    We denote the maximum retrieval level for \(b\in \mathcal{B}\) as \(R(b)\), i.e. the number of targets that are initially positioned below \(b\) within the same stack.
    Observe that the pair \((b,i)\) with \(b \in \mathcal{B}\) and \(i \in \{0, \ldots, R(b)\}\) uniquely identifies the height of \(b\).

    Let \((b,i)\) correspond to position \((s,h)\).
    Let \(A(b,i)\) denote the energy expenditure of a cycle, where every UL in position \((s',h')\) with \(h'\geq h\) is lifted, i.e. 
    \begin{equation}\label{eq:energy1}
        A(b, i)=
        (h(s(b))-h(b) - i)+
        \sum_{s<s(b)}\bigl(h(s)-(h(b)-i-1)).
    \end{equation}
    This amount decreases due to:
    \begin{enumerate}[label=(\roman*)] 
        \item Targets retrieved \emph{before} cycle \(c\) that either lie to the left of \(s(b)\) or above \(b\) in stack \(s(b)\),
            \begin{equation}\label{eq:energy2}
            E_{prev}(b,X_c) =
            \{ b' \in X_c \mid s(b') < s(b) \ \text{or} \ (s(b') = s(b) \ \text{and} \ h(b') > h(b)) \},
            \end{equation}
        \item Targets retrieved \emph{within} cycle \(c\) that are to the left of \(s(b)\) and at height at least \(h(b) - 1 - i\),
            \begin{equation}\label{eq:energy3}
            E_{curr}(b,i,Y_c) =
            | \{b' \in Y_c \mid s(b') < s(b),\ h(b') \ge h(b) - 1 - i \} |.
            \end{equation}
    \end{enumerate}
    Putting everything together, we have
    \(
    E(c) = A(b, i)-E_{prev}(b,X_c)-E_{curr}(b,i,Y_c).
    \) Thus, given a solution, the energy of a cycle can be computed directly from \(A(b, i)\).
    
    Now, we give a representation of information that is sufficient to verify whether a UL is accessible without explicitly encoding stack heights as an array.
    To this end, we define \(U\subseteq T\) as the set of stacks containing at least one target.
    For \(b \in \mathcal{B}\), let \(w(b)\) 
    denote the lowest height of stacks \(t < s(b)\) and \(t \in T\setminus U\).
    Observe that the heights of stacks in \(T\setminus U\) remain unchanged; therefore, \(w(b)\) contains all the relevant information needed to verify its accessibility during the retrieval process.
    Furthermore, we encode the stacks in \(U\) as a height array of length at most \(n\).

    Thus, we define the \emph{sparse encoding} of an instance with
    \(
        \mathcal{I}_{\text{sparse}} = (\mathcal{B}, (h(b),s(b), w(b), h(s(b))_{b\in \mathcal{B}}, \mathcal{A}),
    \)
    where \(\mathcal{A} \coloneqq \{A(b,i) \mid b \in \mathcal{B},\; i \in \{0,\ldots, R(b)\}\).
    Observe that \(R(b) \leq n\), and thus the description length is \(O(n^{2}\log m \log H)\). 
    Based on the previous discussion, to encode a solution, it suffices to specify the sequence \(Y_c\); thus we may encode a solution by the tuple \(\mathcal{X}_{\text{sparse}} = (C,(Y_c)_{c\in C})\).
    \section{Structural Properties and Computational Complexity}
    \label{sec:strucProperties}

    In this section we formalize the problem’s structural properties. We characterize feasibility in Section~\ref{sec:feasibility}, and the geometry of a retrieval cycle in Section~\ref{sec:geometry}. 
    We conclude by establishing a complexity lower bound in Section~\ref{sec:complexity}.
	
	\subsection{Characterizing Feasibility}
	\label{sec:feasibility}

    An instance is infeasible when a target is positioned so high that at least one stack between its stack and the entry is short, and there are too few targets beneath it in its own stack to lower it sufficiently to satisfy accessibility condition~\ref{cond:passage}.
	We formalize this idea with the following lemma.
	\begin{lemma}\label{lem:feasible}
		An instance of the SACRP is feasible if and only if for all \(b \in \mathcal{B}\) and all stacks \(t < s(b)\) we have
		\(
		h(t)\ \ge\ h(b) - R(b).
		\)
	\end{lemma}
	\begin{proof}
		Assume the instance is feasible. Consider a feasible solution and let \(b \in \mathcal{B}\) be retrieved in some cycle \(c\) at step \(i\) from its current position \((s(b), h_c(b))\).
		By accessibility condition \ref{cond:passage}, for every stack \(t < s(b)\), we must have \(d_{c,i-1}(s)=h_c(b)\).
		Since \(d_{c,i-1}(t)\le h_c(t)\le h(t)\), it follows that
		\(
		h(t) \ge\ h_c(b).
		\)
		Moreover, the height of \(b\) can drop before it is retrieved only when target below it in the same stack are removed. Hence
		\(h_c(b)\ge h(b)-R(b)\).
		Putting both inequalities together gives
		\(
		h(t) \ge h(b)-R(b),
		\)
		thereby showing the only-if-claim.
		
		Assume that for every \(b\in\mathcal{B}\) and \(t<s(b)\) we have \(h(t)\ge h(b)-R(b)\).
		We construct a feasible solution creating one cycle for each target \(b \in \mathcal{B}\),
        retrieving targets from right to left and within a stack from bottom to top.
		Consider some target \(b\) in stack \(t\) at the moment it is to be retrieved in cycle \(c\). 
		Observe that the number of target below \(b\) in stack \(t\) that have already been retrieved is by construction \(R(b)\) and then current height of \(b\) is \(h_c(b)=h(b)-R(b)\).
		Because stacks to the left have not yet been processed under this schedule, for each \(s<t\) we have \(h_c(s)=h(s)\).
		By the assumption, we have 
		\(
		h_c(s)=h(s) \ge h(b)-R(b)  = h_c(b).
		\)
		We choose the clearance levels in cycle \(c\) such that, for \(t = s(b)\), we have \(\ell_c(s(b)) = h_c(b)+1 \);
		for \(t < s(b)\), we have
		\(\ell_c(t) = h_c(b)\); and for \(t >s(b)\), we have \(\ell(t) = h_c(t)\).
		Observe that these are valid since \(\ell_c(s)\le h_c(s)\) for all \(s\) and \(b\) is accessible.
		Iterating this argument over all targets completes a feasible retrieval sequence, proving the if-claim.
	\end{proof}
	
	The procedure used as an argument in the previous proof gives us an algorithm to compute a feasible solution with runtime that is dominated by sorting the ULs, which can be done in \(O(n\log n )\) time. 
	We can thus conclude that identifying and constructing feasible instances of SACRP is easy and will assume that input instances are feasible throughout the paper.
	
	\subsection{Characterizing Retrieval Cycles}
	\label{sec:geometry}

	Consider an SACRP instance and a feasible solution. 
	Consider some cycle \(c\) and consider the sequence \(Y_c\).
    We denote by \(Y\) the set of targets in \(Y_c\) and refer to this set as a \emph{batch}.
	We propose a number of technical lemmas that characterize the shape of \(Y\).
	
	We list the distinct heights \(H\coloneqq\{h_c(b)\mid b\in Y\}\) in increasing order as \(h_1<\cdots<h_p\) and, for all \(i=1,\ldots,k\), define
	\(
	s_i \coloneqq \max\{s(b)\mid b\in Y,\ h_c(b)=h_i\},
	\)
	i.e., the maximum stack index of any target retrieved at height \(h_i\).
	
	\begin{lemma} \label{lem:continuous}
		It holds that \(h_i = h_{i+1} + 1\) for all \(i\), i.e., the heights of \(Y\) form a consecutive interval. In other words, \(Y\) is height-continuous.
	\end{lemma}
	\begin{proof}
		Assume otherwise, and consider some \(b \in Y\) such that no target is retrieved at height \(h_c(b) + 1\), while a target \(b' \in Y\) is retrieved at height \(h_c(b') >  h_c(b) + 1\).
		Observe that to retrieve \(b'\), accessibility condition \ref{cond:passage} implies that there are ULs at height \(h_c(b) + 1\) in stacks \(1, 2, \ldots, s(b')\) that are not lifted.
		However, to retrieve \(b\), accessibility conditions~\ref{cond:above} and~\ref{cond:passage} require that at height \(h_c(b)+1\) in stacks \(s \le s(b)\), ULs are either absent or lifted.
		This yields a contradiction and proves the claim.
	\end{proof}
	
	\begin{lemma} \label{lem:unimodal}
		There exists an index \(\ell\in\{1,\ldots,k\}\) such that \(s_1 \le s_2 \le \cdots \le s_\ell \ge s_{\ell+1} \ge \cdots \ge s_k\). In other words, \(Y\) is stack-unimodular.
	\end{lemma}
	\begin{proof}
		Assume otherwise and consider some  \(i,j,k \in \{1,2,\ldots, p\}\), such that \(h_i < h_j < h_k\) and \(s_i > s_j\) and \(s_k > s_j\). 
		Let \(b_i, b_j, b_k\in Y\) denote the targets at positions \((h_i,s_i)\), \((h_j,s_j)\) and \((h_k,s_k)\), respectively.
		
		To retrieve \(b_i\), accessibility condition \ref{cond:passage} implies that all ULs at height \(h_j\) in stacks \(s_j, s_j + 1\ldots, s_i\) are either absent or lifted. 
		Note that these ULs are not retrieved because, by assumption, they are not in \(Y\).
		To retrieve \(b_k\), accessibility conditions \ref{cond:above} and \ref{cond:passage} imply that ULs at height \(h_j\) in stacks \(s_j, s_j + 1\ldots, s_k\) are neither absent nor lifted.
		This yields a contradiction and proves the claim.
	\end{proof}
	
	\begin{lemma} \label{lem:prefix}
		Let \(\ell\) be defined as in Lemma~\ref{lem:unimodal}. 
		For each \(1 \le i < \ell\), every UL that lies at height \(h_{i+1}\) and in stacks \(1, 2, \ldots, s_i\) belongs to \(Y\) or is absent.
		For each \(\ell < i \le k\), every UL that lies on height \(h_{i-1}\) and in stacks \(1, 2, \ldots, s_i\) belongs to \(Y\) or is absent.
		In other words, \(Y\) is prefix-closed.
	\end{lemma}
	\begin{proof}
        Assume otherwise and consider some \(i \in \{1, 2, \ldots, \ell-1 \}\)
		Let \(1 \le i < \ell\) and assume there exist a UL at position \((h_{i+1},t)\) with \(t \le s_i\) that does not belong to~\(Y\).    
		By construction, we know that there exists a target \(b\in Y\) with \(s(b) \geq  t\) on height \(h_i\).
		To retrieve \(b\), accessibility condition \ref{cond:passage} implies that ULs on height \(h_{i+1}\) in stacks \(1, 2, \ldots, {t}\) are retrieved earlier, i.e. in \(Y\), or are absent.
		This yields a contradiction and shows the claim.
		The remaining case \(\ell < i \le k\) is analogous.
	\end{proof}
	
	\begin{lemma} \label{lem:batchGeom}
		Consider a batch \(Y\) and assume that for every \(b\in Y\) and every stack \(t<s(b)\) we have \(h(t) \ge h(b)\).
		If \(Y\) is height-continuous, stack-unimodular, and prefix-closed, then there exists a sequence, such that every target UL is accessible. In other words, \(Y\) is feasible.
	\end{lemma}
	\begin{proof}
		We say a stack is \emph{empty} if it does not contain any target from \(Y\). For each stack \(t\), we define the lift thresholds as follows. 
		If \(t\) is non-empty, let \(\ell_c(t)\coloneqq \{ h_c(b) + 1 \mid b \in s(b)\}\), i.e., lift everything above the topmost target UL in stack \(t\).  
		If \(t\) is empty and there exists a target UL in stacks to its right, let \(\ell_c(t)\coloneqq \max \{ h_c(b) \mid b\in Y,\ s(b)>t\}\), i.e., lift all ULs above height \(\ell_c(t)\), which equals the maximum residual height of any target UL in \(Y_c\) to the right of \(t\).  
		Otherwise, set \(\ell_c(t)=h_c(t)\), i.e., do not lift any ULs from stack \(t\).
		
		We define the retrieval sequence as follows. 
        We retrieve the targets from top to bottom, and within each height from left to right.
        
		Assume we are at an arbitrary step \(i\) of the retrieval process and let \(b\) denote the next target to be retrieved.
		We claim that \(b\) is accessible.
		To see this, assume that \(h(b) = \max H\).
		By construction, we know that accessibility condition~\ref{cond:above} is satisfied.
		Furthermore, again by construction, we know that every empty stack to the left of \(b\) has a height of at least \(h(b)\).
		Furthermore, every target UL in \(Y_c\) at height \(h(b)\) and to the left of \(b\) is retrieved in an earlier step.
		Therefore, accessibility condition~\ref{cond:passage} is satisfied and \(b\) is accessible.
		Assume that \(h(b) < \max H\).
		Because \(Y_c\) is height-continuous, stack-unimodular, and prefix-closed, we know that the UL at position \((s(b), h(b)+1)\) is either lifted, already retrieved, or does not exist.
		Thus, accessibility condition~\ref{cond:above} is satisfied.
		Furthermore, by prefix-closure and our choice of the retrieval sequence, we know that every UL at position \((t, h(b))\) for \(t < s(b)\) is either lifted, already retrieved or does not exist.
		Therefore, accessibility condition~\ref{cond:passage} is satisfied, and \(b\) is accessible.
	\end{proof}

    We call a batch \emph{weakly triangular} if it is height-continuous, stack-unimodular, and prefix-closed.
	Lemmas~\ref{lem:continuous} to \ref{lem:batchGeom} imply that choosing a lifting configuration and retrieval sequence reduces to selecting weakly triangular batches. 
    Thus, we may restrict the solution space so that each cycle consists of a weakly triangular batch.
    To give the reader more intuition, we depict some weakly triangular batches in Figure~\ref{fig:geometry}.
    
\begin{figure}[t]
    \centering
    \begin{subfigure}[t]{0.42\textwidth}
        \centering
        \includegraphics[page=1,width=.7\linewidth]{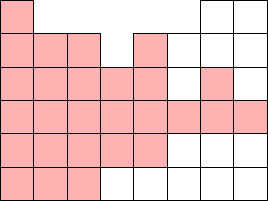}
        \caption{Example 1: weakly triangular batch \(Y\).}
        \label{fig:geometry1}
    \end{subfigure}\hfill
    \begin{subfigure}[t]{0.42\textwidth}
        \centering
        \includegraphics[page=2,width=.7\linewidth]{geometry.pdf}
        \caption{Example 2: weakly triangular batch \(Y\).}
        \label{fig:geometry2}
    \end{subfigure}

    \vspace{0.6em} 

    \begin{subfigure}[t]{0.42\textwidth}
        \centering
        \includegraphics[page=3,width=.7\linewidth]{geometry.pdf}
        \caption{Example 3: weakly triangular batch \(Y\).}
        \label{fig:geometry3}
    \end{subfigure}\hfill
    \begin{subfigure}[t]{0.42\textwidth}
        \centering
        \includegraphics[page=4,width=.7\linewidth]{geometry.pdf}
        \caption{Example 4: weakly triangular batch \(Y\).}
        \label{fig:geometry4}
    \end{subfigure}

    \caption{Illustration of four examples of weakly triangular batches \(Y\). Targets included in \(Y\) are highlighted in red.}
    \label{fig:geometry}
\end{figure}

	\subsection{Strong NP-Completeness of the SACRP}
	\label{sec:complexity}
	
	Our next result shows that it is unlikely that there exists a polynomial-time algorithm that computes an optimal solution. 
	The underlying reason is that removing lower targets may align the heights of targets that lie above, so targets that initially cannot be retrieved in one cycle can be retrieved in a single cycle.
	Thus, there are cases in which we may make seemingly costly decisions early on to create a specific structure, which later proves to be cost-efficient to clear.
	The decision of which targets to group together has a strong connection to partition problems, which is shown in the following result.
	\begin{theorem}
		The SACRP is strongly NP-complete.
	\end{theorem}
	\begin{proof}
		Under either the dense or the sparse encoding, a candidate solution together with an integer bound \(K\) is a polynomial-size certificate: by simulating the cycles, we can in polynomial time verify (i) that each target is accessible when retrieved (conditions \ref{cond:above} and \ref{cond:passage}), (ii) that every target is retrieved exactly once, and (iii) that the energy expenditure \(E(\mathcal{X})\) satisfies \(E(\mathcal{X}) \le K\). Hence, SACRP belongs to NP.
		
		In what follows, we reduce the 3-Partition Problem, which is well-known to be strongly NP-complete (see Garey and Johnson \cite{garey1979computers}) to the SACRP.
		Here, we are given a multiset of \(3m\) positive integers \(A = \{a_1, a_2, \dots, a_{3m}\}\) and an integer \(T\), with the constraints \(T/4 < a_i < T/2\) for each \(1 \leq i \leq 3m\) and \(\sum_i a_i = mT\). 
		We may assume that all numbers are even.
		We consider the decision variant of the 3-Partition Problem, which asks whether there exists a partition of \(A\) into \(m\) triplets \(A_1, A_2, \dots, A_m\) such that the sum of each subset is \(T\).
		
		Consider an instance \(\mathcal{I}_P\) of the 3-Partition Problem with integers \(A = \{a_1, a_2, \dots, a_{3m}\}\) and target sum \(T\).
		We construct an instance \(\mathcal{I}_R\) of the SACRP as follows (see Figure~\ref{fig:complexity} for an illustration).
		We formulate the instance using the dense encoding scheme, as it is easier to follow; however, an analogous reduction can be constructed using the sparse encoding scheme.
		Let \(L\) be a large integer to be defined later, and let the slice consist of \(3\) stacks with the following heights:
		\[
		h(t) =
		\begin{cases} 
			m(2T + 1) + n + L & \text{if } t = 1, \\
			mT + n  & \text{if } t = 2, \\
			m(T+1) + n  & \text{if } t = 3. \\
		\end{cases}
		\]
		The picklist is partitioned into three sets \(X\), \(Y\), and \(Z\). 
		Set \(X\) contains \(mT\) targets and is partitioned into \(3m\) subsets \(X_1, \ldots, X_{3m}\). 
		Furthermore, each set \(X_i\) corresponds to an integer \(a_i\in A\) and is defined by \(X_i = \{x_{i1}, \ldots, x_{ia_i}\}\).
		In the first stack and at heights \(1\) to \(m(T + 1)\), each \(X_i\) occupies a consecutive substack, and every pair of consecutive \(X_i\) and \(X_{i+1}\) is separated by a single non-target.
		Set \(Y=\{y_1,\ldots,y_m\}\) contains \(m\) elements.  
	 	Element \(y_1\) lies \(T + 2\) heights above the topmost element of \(X\), and successive elements in \(Y\) are separated by \(T\) non-target ULs. 
		Set \(Z=\{z_1,\ldots,z_m\}\) contains \(m\) elements and forms a consecutive substack in stack \(3\), i.e., each \(z_i\) lies directly above \(z_{i-1}\) for \(i > 1\).
		Moreover, \(z_1\) lies two heights above the topmost element of \(X\) and \(T\) heights below the bottommost element of \(Y\).
		
		Let \(K \coloneqq  m(4T +  2 ) + 3n + L \), which is the total number of ULs in the slice, and we define \(L \coloneqq 4m(m(4T +  2 ) + 3n) \).
		Observe that we have \(n = mT + 2m\). 
		Thus, the reduction runs in pseudo-polynomial time. 
		We claim that \(\mathcal{I}_P\) is a {\sc Yes}-instance if and only if there exists a solution to \(\mathcal{I}_R\) with total energy expenditure at most \(L (4m + 1)\).
		In the remainder of the proof, we show this claim.
		
		First, we claim that a solution to \(\mathcal{I}_R\) has energy at most \(L (4m + 1)\) if and only if it retrieves all targets in \(4m\) cycles. 
		To see this, consider a solution \(\mathcal{X}\) that clears all targets in \(4m\) cycles.
		Observe that \(K\) is an upper bound on the energy expenditure per cycle.
		Thus, we know that 
		\[E(\mathcal{X})\leq 4m\cdot K = 4m(m(4T +  2 ) + 3n) + 4m\cdot L \le  L(4m+1),
		\]
		where we used the upper bound on the energy expenditure on a cycle in the first inequality, and the second inequality follows from the choice of \(L\).
		Conversely, consider a solution \(\mathcal{X}\) that has at most \(L(4m+1)\) energy expenditure.
		Observe that, as there are \(L\) ULs on top of the highest target in the first stack, every cycle requires at least \(L\) energy.
		Thus, if \(k\) is the number of cycles in \(\mathcal{X}\), we have \(kL\le E(\mathcal{X})\). Together with our assumption that
		\(E(\mathcal{X}) \le L(4m+1)\), we know that \(k \le 4m+1\), showing the claim. 
		
		We claim that the targets in sets \(X\) and \(Y\) can be retrieved in at least \(4m\) cycles.
		To see this, observe that in the first stack, each element in \(X\) and \(Y\) appears as \(4m\) pairwise disjoint vertical segments, each separated from the next by at least one non-target. 
		Thus, Lemma~\ref{lem:continuous} implies that we require at least \(4m\) cycles, hence, the claim follows.
		
		Similarly, we claim that \(Z\) can be retrieved in at least \(m\) cycles.
		Indeed, as there are no targets in the second stack, Lemma~\ref{lem:prefix} implies that only one element from \(Z\) can be retrieved in a cycle.
		
		Now, consider a solution of \(I_R\) in which all targets are retrieved within \(4m\) cycles.
		Based on our previous discussion, we know that there must be a cycle in which \(y_i\) and \(z_i\) are retrieved for all \(i \in \{1, \ldots, m\}\).
		Consider the initial layout. We aim to retrieve \(y_1\) and \(z_1\) in a single cycle. 
		Observe that we cannot lower \(z_1\), and therefore we have to lower \(y_1\) to height \(h(z_1)\) or \(h(z_1)-1\), which enables us to take both targets in one cycle.
		This arises precisely when we have to retrieve  sets  from \(X\) with \(T\) or \(T+1\) elements. 
		From the previous claims, we know that these sets must form a triplet \(X_i, X_j, X_k\) with \(T\) elements.
		Thus, the corresponding numbers \(a_i, a_j, a_k\) also form a triplet that sums to \(T\).
		After \(X_i, X_j, X_k\), \(y_1\) and \(z_1\) are retrieved.
		The element \(y_2\) is, again, \(T\) heights above \(z_2\), so the same arguments apply.
		Thus, we can use the triplets of \(X\) to iteratively build a solution to \(I_T\), thereby showing the only-if-claim.
		
		Now, consider a solution of \(I_T\).
		For any triple \(A_{i}\) with subset sum equal to \(T\), retrieve in three cycles the three sets of \(X\) corresponding to numbers in \(A_i\).
		Doing so lowers the first stack by exactly \(T\) levels, so that the lowest remaining target in \(Y\) aligns with the lowest remaining target in \(Z\); these two can therefore be retrieved in the same cycle.  
		Repeating this procedure for all \(m\) triples clears every set of \(X\) and retrieve all target from \(Y\) and \(Z\) in a total of \(4m\) cycles; three cycles per triple to clear the corresponding set of \(X\), followed by one cycle for the pair \((y_i, z_i)\). This show the if-claim.
		\begin{figure}
			\centering
			\begin{subfigure}[b]{.47\textwidth}
				\centering
				\includegraphics[scale = 1, page=1]{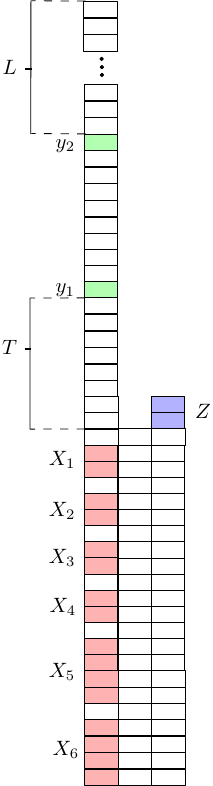}
				\caption{Initial construction of the instance. Observe that the vertical distance between \(y_1\) the first element in \(Z\) is exactly target sum \(T\).}
				\label{fig:sub1}
			\end{subfigure}%
			\hspace{15pt}
			\begin{subfigure}[b]{.47\textwidth}
				\centering
				\includegraphics[scale = 1, page=2]{complexity.pdf}
				\caption{Construction, after sets \(X_1\), \(X_3\) and \(X_6\) have been retrieved in three cycles. Observe that \(y_1\) is aligned with the first element of \(Z\), that is \(z_1\). Thus, \(y_1\) and \(z_1\) can now be retrieved in a single cycle.}
				\label{fig:sub2}
			\end{subfigure}
			\caption{Illustration of the construction for the following instance \(\mathcal{I}_P\): \(S=\{2,2,2,2,4,4\}\), thus \(m=2\), and \(T=8\). There is solution to \(\mathcal{I}_P\) with \(A_i = \{2,2,4\}\) and \(A_2 = \{2,2,4\}\). The red ULs corresponds to the set \(X\), the green UL corresponds to the set \(Y\) and the blue ULs corresponds to the set \(Z\).}
			\label{fig:complexity}
		\end{figure}
	\end{proof}
    
	\section{Mathematical Programming Model Formulation}\label{sec:mathProg}
	In this section, we propose a MIP formulation for the SACRP. 
    We first develop a local characterization of weakly triangular batches in Section~\ref{sec:batching} and then incorporate it into the model described in Section~\ref{sec:mip}.
        
    \subsection{Local Characterization of Weakly Triangular Batches}\label{sec:batching}

    The general idea of the mathematical model is the following. 
    In each cycle, a batch \(Y\) is initiated by activating a target \(b(Y)\).
    We say \(b(Y)\) is \emph{anchoring} \(Y\).
    We allow the model to iteratively add targets to \(Y\) while maintaining \(b(Y)\) as the rightmost target and preserving the weakly triangular structure of \(Y\).
    The local decision to add a target \(b\) to an existing batch \(Y\) is defined in the following lemmas.

    Specifically, fix a cycle \(c\) and consider the set of targets that are not retrieved and accessible at the beginning of \(c\), denoted with \(X_c\).
    Fix a weakly triangular batch \(Y\) and some \(b(Y)\), which is a rightmost target in \(Y\).
    Let \(s(Y) \coloneqq s(b(Y))\) and \(h(Y) \coloneqq h(b(Y))\).
    Consider some \(b \in X_c\) with \(s(b) \leq s(Y)\) and let \(h = h(b)\), \(s = s(b)\) and \(Y' \coloneqq Y \cup \{b\}\).
    Denote the energy expenditure in cycle \(c\) of a weakly triangular batch \(Z\) with \(A(Z)\).

    \begin{lemma}\label{lem:local1}
        If \(h = h(Y)\) and \(s \le s(Y)\), then \(Y'\) is weakly triangular. 
        Furthermore, we have \(A(Y') = A(Y) - 1\).
    \end{lemma}
    \begin{proof}
        Observe that the set of distinct heights and rightmost stack in each height remains unchanged; hence, \(Y'\) remains height-continuous and stack-unimodular.
        Observe that \(b\) is on height \(h\) and \(Y\) is prefix-closed, we know that there are no targets in \(Y\) on heights higher than \(h\). 
        Hence, \(Y'\) remains prefix-closed.
        Furthermore, observe that, retrieving \(b\) within the same cycle, implies that there is one less UL that needs to be elevated. Therefore, we have that \(A(Y') = A(Y) - 1\).
    \end{proof}

    \begin{lemma}\label{lem:local2}
       If every position \((h-1, s')\) for all \(s' \le s\) is occupied by a target \(b'\in X_c\) and \(b' \in Y\), then \(Y'\) is weakly triangular. 
        Furthermore, we have \(A(Y') = A(Y) - 1\).
    \end{lemma}
    \begin{proof}
        Observe that by assumption, there exist \(b' \in Y\) with height \(h-1\), thus \(Y'\) is height-continuous.
        By construction \(s(b) < s(Y)\), thus \(Y'\) is stack-unimodular.
        By assumption, every position at height \(h-1\) and to the left of \(b\) is in \(Y\), thus \(Y'\) is prefix-closed.
        Furthermore, observe that, retrieving \(b\) within the same cycle, implies that there is one less UL that needs to be elevated. Therefore, we have that \(A(Y') = A(Y) -~1\).
    \end{proof}

    \begin{lemma}\label{lem:local3}
       If \(h = h(Y)-1\) or if position \((h+1, s)\) is occupied by a target \(b'\in X_c\) and \(b' \in Y\), and if \(b\) is in the first stack or the position \((h,s-1)\) is occupied by a target \(b'\in X_c\) and \(b' \in Y\), then \(Y'\) is weakly triangular. 
    Furthermore, we have \(A(Y') = A(Y)\).
    \end{lemma}
    \begin{proof}
        Observe that by assumption, there exist \(b' \in Y\) with height \(h+1\), thus \(Y'\) is height-continuous.
        By construction \(s(b) < s(Y)\), thus \(Y'\) is stack-unimodular.
        By assumption, every position at height \(h+1\) and to the left of \(b\) is in \(Y\), thus \(Y'\) is prefix-closed.
       Furthermore, observe that retrieving \(b\) within the same cycle, implies that the same ULs need to be elevated. Therefore, we have that \(A(Y') = A(Y)\).
    \end{proof}
    \begin{lemma}\label{lem:completeness}
		If \(Y\) is weakly triangular, then there exists an anchoring target \(b(Y)\) and an ordering \(y_1,\ldots,y_{|Y|}\) of \(Y\) such that, for every \(i\in\{1,\ldots,|Y|-1\}\), the element \(y_{i+1}\) satisfies one of Lemmas~\ref{lem:local1}--\ref{lem:local3} with respect to \(Y_i=\{y_1,\ldots,y_i\}\).
	\end{lemma}
	
	\begin{proof}
		Given a batch \(Y\) that is weakly triangular.
        Let \(b_1\) denote an target in \(Y\) in the rightmost stack, ties are broken arbitrarily.
        We construct a batch \(Y'\) by iteratively applying Lemmas~\ref{lem:local1}--\ref{lem:local3}.
		Initiate \(Y' \coloneqq\{b_1\}\) and \(b(Y') \coloneqq b_1\)
		For targets on height \(h(Y')\), we order them from right to left and use Lemma~\ref{lem:local1} to, iteratively, add them to \(Y'\).
		For targets with height more than \(h(Y')\), we order them non-decreasing in their heights and from left to right inside a height. 
		We then use Lemma~\ref{lem:local2} to, iteratively, add them to \(Y'\).
		For targets with height less than \(h(Y')\), we order them non-increasing in their heights and from left to right inside a height.
		We then use Lemma~\ref{lem:local3} to, iteratively, add them to \(Y'\).
		It is easy to see, that \(Y' = Y\), hence showing the claim.
		\end{proof}

	\subsection{Model Formulation}\label{sec:mip}
	
	We now define the decision variables of the model. 
	Each decision variable is defined for a cycle. 
    Observe that any feasible solution consists of at most \(n\) cycles. 
    If the optimal solution uses fewer than \(n\) cycles, then some of the cycles are empty and can be removed in a post-processing step without changing the energy expenditure.
    Thus, we define all decisions variables for cycle \(c\) for all \(1 \leq c \leq n\).
	Furthermore, to ensure that the model is compact, we represent the position of target \(b\) by the pair \((b,i)\), where \(i\) is the retrieval level.
	We define
	\begin{equation*}
		x(c,b,i) \coloneqq
		\begin{cases}
			1 & \parbox[t]{6.5cm}{if \(b\) is retrieved at retrieval level \(i\) in cycle~\(c\),}\\[8pt]
			0 & \text{else,}
		\end{cases}
		\quad \forall \, 1 \leq c \leq n,\; b\in \mathcal{B},\; 1 \leq i \leq R(b),
	\end{equation*}
	\begin{equation*}
		y(c,b,i) \coloneqq
		\begin{cases}
			1 & \parbox[t]{6.5cm}{if cycle \(c\) is anchored by \(b\) at retrieval level~\(i\),}\\[8pt]
			0 & \text{else,}
		\end{cases}
		\quad \forall \, 1 \leq c \leq n,\; b\in \mathcal{B},\; 1 \leq i \leq R(b),
	\end{equation*}
	
	\begin{equation*}
		z_1(c,b,i) \coloneqq
		\begin{cases}
			1 & \parbox[t]{6.5cm}{if \(b\) is retrieved at retrieval level \(i\) in cycle~\(c\) according to  Lemma~\ref{lem:local1},}\\[8pt]
			0 & \text{else,}
		\end{cases}
		\quad \forall \, 1 \leq c \leq n,\; b\in \mathcal{B},\; 1 \leq i \leq R(b),
	\end{equation*}
	
	\begin{equation*}
		z_2(c,b,i) \coloneqq
		\begin{cases}
			1 & \parbox[t]{6.5cm}{if \(b\) is retrieved at retrieval level \(i\) in cycle~\(c\) according to  Lemma~\ref{lem:local2},}\\[8pt]
			0 & \text{else,}
		\end{cases}
		\quad \forall \, 1 \leq c \leq n,\; b\in \mathcal{B},\; 1 \leq i \leq R(b),
	\end{equation*}
	
	\begin{equation*}
		z_3(c,b,i) \coloneqq
		\begin{cases}
			1 & \parbox[t]{6.5cm}{if \(b\) is retrieved at retrieval level \(i\) in cycle~\(c\) according to  Lemma~\ref{lem:local3} and \(b\) is one height lower then the anchoring variable,}\\[8pt]
			0 & \text{else,}
		\end{cases}
		\quad \forall \, 1 \leq c \leq n,\; b\in \mathcal{B},\; 1 \leq i \leq R(b),
	\end{equation*}
	\begin{equation*}
		z_4(c,b,i) \coloneqq
		\begin{cases}
			1 & \parbox[t]{6.5cm}{if \(b\) is retrieved at retrieval level \(i\) in cycle~\(c\) according to  Lemma~\ref{lem:local3} and \(b\) is more than one height lower then the anchoring variable,}\\[8pt]
			0 & \text{else,}
		\end{cases}
		\quad \forall \, 1 \leq c \leq n,\; b\in \mathcal{B},\; 1 \leq i \leq R(b),
	\end{equation*}

	\begin{equation*}
		u(c,b,i) \coloneqq
		\begin{cases}
			1 & \parbox[t]{6.5cm}{if \(b\) is at retrieval level \(i\) in the beginning of cycle \(c\)}\\[8pt]
			0 & \text{else,}
		\end{cases}
		\quad \forall \, 1 \leq c \leq n,\; b\in \mathcal{B},\; 1 \leq i \leq R(b).
	\end{equation*}

     Additionally, we define a decision variable \(E(c) \ge 0\) for each cycle \(c=1,\ldots,n\), representing the energy expenditure in cycle \(c\). The variable \(h(c,t) \ge 0\) denotes the height of stack \(t \in U\) at the beginning of cycle~\(c\).
    We use the notation \(A(b,i)\), \(U\), and \(w(b)\) as introduced in Section~\ref{sec:sparse}. Finally, define \(f(b)\) as a binary indicator of whether \(b\) lies in the first stack: \(f(b)=1\) if \(s(b)=1\), and \(f(b)=0\) otherwise.
	
	\begin{figure}[htbp]
		\centering
		\resizebox{\textwidth}{!}{%
			\fbox{%
				\parbox{\textwidth}{%
					\setlength{\jot}{10pt}
					\allowdisplaybreaks
					\scriptsize
					\begin{equation}
						\text{Minimize} \quad \sum_{c} E(c)\label{eq:obj}
					\end{equation}
					
					\text{Subject to:}
					
					\begin{alignat}{2}
						&\quad \sum_{c}\sum_{i} x(c,b,i) = 1 && \quad \forall b, \label{eq:const1}\\
						&\quad h(1,s) = h(s) && \quad \forall s, \label{eq:const2}\\
						&\quad h(c,s)-\sum_{\{(b,i): s(b)=s\}}x(c,b,i)=h(c+1,s)&&\quad \forall s,\; 1 \leq c < n, \label{eq:const3}\\
						&\quad \sum_{c'< c}\sum_{\substack{\{(b',i'):\\ h(b')-i'<h(b)-i,\\ s(b')=s(b)\}}}x(c',b',i')=\sum_{i} i\cdot u(c,b,i) &&\quad\forall c,\;b, \label{eq:const4}\\
						&\quad x(c,b,i)\leq u(c,b,i)&&\quad\forall c,\;b,\;i, \label{eq:const5}\\
						&\quad \sum_{i} u(c,b,i)=1&&\quad\forall c,\; b, \label{eq:const6}\\
						&\quad h(b)-i-1\leq w(b)+M(1-x(c,b,i))&&\quad\forall c,\;b,\;i, \label{eq:const7}\\
						&\quad h(b)-i-1\leq h(c,s)+M(1-x(c,b,i))&&\quad\forall c,\;b,\;i,\;1\leq s<s(b), \label{eq:const8}\\
						&\quad z_1(c,b,i)\leq\sum_{\substack{\{(b',i'):\\ h(b)-i=h(b')-i',\\ s(b)<s(b')\}}}y(c,b',i') &&\quad\forall c,\;b,\;i, \label{eq:const9}\\
						&\quad 
						s(b)\cdot (z_2(c,b,i)-1)\leq\sum_{\substack{\{(b',i'):\\ h(b)-i-1=h(b')-i',\\ s(b)\leq s(b')\}}}(x(c,b',i')-z_3(c,b',i')-z_4(c,b',i'))-s(b)&&\quad\forall c,\;b,\;i, \label{eq:const10}\\
						&\quad
						2z_3(c,b,i) - f(b)\leq 
						\sum_{\substack{\{(b',i'):\\ h(b) - i = h(b') - i' - 1,\\ s(b) \leq s(b')\}}} y(c,b',i') 
						+ \sum_{\substack{\{(b',i'):\\ h(b) - i = h(b') - i',\\ s(b) = s(b') + 1\}}} x(c,b',i')
						&&\quad\forall c,\;b,\; i \label{eq:const11}\\ 
						&\quad
						2z_4(c,b,i) - f(b) \leq 
						\sum_{\substack{\{(b',i'):\\ h(b) - i = h(b') - i' - 1,\\ s(b) = s(b')\}}} x(c,b',i') 
						+ \sum_{\substack{\{(b',i'):\\ h(b) - i = h(b') - i',\\ s(b) = s(b') + 1\}}} x(c,b',i') 
						&&\quad\forall c,\; b,\; i \label{eq:const12}\\ 
						&\quad x(c,b,i)\leq z_1(c,b,i)+z_2(c,b,i)+z_3(c,b,i)+z_4(c,b,i)&&\quad\forall c,\;b,\;i, \label{eq:const13}\\
						&\quad \sum_{i}y(c,b,i)\leq 1&&\quad\forall c,\;b, \label{eq:const14}\\
						&\quad x(c,b,i)\leq \sum_{\{(b',i'): s(b')\leq s(b) \}}y(c,b',i')&&\quad\forall c,\;b,\;i  \label{eq:const15}\\
						&\quad E(c)\geq A(b,i)\cdot y(c,b,i)-\sum_{c'\leq c}\sum_{\substack{\{(b',i'): \\ s(b)>s(b')\\ \text{or } (s(b)=s(b')\\ \text{and }h(b)<h(b'))\}}}x(c',b',i')+\sum_{b',\;i'} (z_3(c,b',i')+z_4(c,b',i'))&&\quad\forall c,\;b,\;i  \label{eq:const16}\\
						&\quad x(c,b,i),\;y(c,b,i),\;z_1(c,b,i),\;z_2(c,b,i),\;z_3(c,b,i),\;z_4(c,b,i),\;u(c,b,i)\in\{0,1\}&& \quad\forall c,\;b,\;i, \label{eq:const17}\\
						&\quad E(c),\; h(c,s) \geq 0 && \quad \forall c,\; s.\label{eq:const18}
					\end{alignat}
		}}}
		\captionsetup{labelformat=empty} 
		\caption{\textbf{Model 1}: Mathematical programming model for the SACRP.} 
		\label{model:model1}
	\end{figure}
	
	The mathematical model formulation is given in Model 1. 
	The objective function in Equation~\eqref{eq:obj} aims to minimize energy expenditure. 
	Constraints~\eqref{eq:const1} ensure that each target is retrieved exactly once. 
	Heights are initialized and updated in Constraints~\eqref{eq:const2} and~\eqref{eq:const3}.
	Constraints~\eqref{eq:const4} to~\eqref{eq:const6} model the changes in the retrieval level of targets.
	Constraints~\eqref{eq:const7} ensure that a target can only be retrieved if it is accessible with respect to stacks \(T \setminus U\), i.e., the stacks that do not contain any target.
    Constraints~\eqref{eq:const8} ensure that a target can only be retrieved if it is accessible with respect to stacks \(U\).
	This constraints require a sufficiently large constant \[M \coloneqq \max_b h(b) - \min_b \{w(b), h(s(b))\}.\]
    Constraints~\eqref{eq:const9} ensure that a target can be retrieved according to Lemma~\ref{lem:local1}.
    Constraints~\eqref{eq:const10} ensure that a target can be retrieved according to Lemma~\ref{lem:local2}.
	Constraints~\eqref{eq:const11} and~\eqref{eq:const12} capture Lemma~\ref{lem:local3}.
	In Constraints~\eqref{eq:const14}, we ensure that there is at most one anchoring target for each cycle.
	Finally, Constraints~\eqref{eq:const15} 
    ensure that the anchoring target is rightmost in the batch.
	In Constraints~\eqref{eq:const16}, we model the energy expenditure with Equation~\eqref{eq:energy1} to~\eqref{eq:energy3}.
	Finally, in Constraints~\eqref{eq:const17} to Constraints~\eqref{eq:const18}, we define the domains of the decision variables.

    Based on Lemma~\ref{lem:local1} to~\ref{lem:completeness}, the following result immediately follows.
    
	\begin{theorem}
        Every optimal solution to Model~1 corresponds to an optimal solution of the SACRP, and vice versa. The formulation uses \(O(n^3)\) variables and \(O(n^3)\) constraints.
	\end{theorem}
	
	\section{Algorithms}\label{sec:algorithms}
    In this section, we present two approaches for the SACRP: a dynamic programming algorithm (Section~\ref{sec:DP}) and a greedy greedy algorithm (Section~\ref{sec:greedy}).
    
	\subsection{Dynamic Programming Approach}\label{sec:DP}
	
    Consider an arbitrary ordering of the targets \(\mathcal{B} = \{b_1, \ldots, b_n \}\).
	We develop a dynamic programming approach (DP) that uses the state space 
	\(X = (x_1, x_2, \ldots, x_n)\), where each \(x_i \in \{0,1\}\) indicates whether target \(b_i\) has been picked (\(x_i = 1\)) or not (\(x_i = 0\)).    
	We give an algorithm that constructs the transitions between states.

	For each state \(X_s \in \{0,1\}^n\), initiate a batch with \(Y \leftarrow \emptyset\) and let \(X \leftarrow X_s\).
    Let \(B(X,Y)\) denote the set of accessible targets \(b\) such that (i) \(b\) is not yet retrieved and not in \(Y\), (ii) \(b\) lies to the right of and not above every element of \(Y\), and (iii) \(Y' \coloneqq Y \cup \{b\}\) is weakly triangular.
    If \(b(X,Y)\) is non-empty, then for each \(b_i\in B(X,Y)\) define the successor state \(X_t\) from \(X\) by setting \(x_i = 1\).
    We construct a transition between \(X_s\) and \(X_t\) with cost \(A(Y')\).
    We let \(X \leftarrow X_t\) and \(Y \leftarrow Y'\) and repeat.
    If \(b(X,Y)\) is empty, we stop the recursion.
	Given the complete state space and transitions, we compute a shortest path from state \((0, \ldots, 0)\) to \((1, \ldots, 1)\) and return its cost. We refer to this as Algorithm 1.
	
	\begin{theorem}
		Algorithm 1 computes an optimal solution to the SACRP in \(3^{O(n)}\) time.
	\end{theorem}
	\begin{proof}
		We discuss the time complexity and the correctness of the algorithm separately.
		\medskip 
		
		\noindent \emph{Time Complexity.} Given the structure of the state space, states can be grouped to \(n+1\) stages such that in each stage \(k \in \{0, \ldots, n\}\) it holds for any state \(X\) that \(\sum_{i=1}^n x_{i} = k \).
		Observe that stage \(k\) has \(\binom{n}{k}\) states and there are at most \(2^{n-k} \) transitions from every such state to states in stage \(\ell\) with \(\ell > k\). Therefore, the size of the state space is bounded by \(\sum_k \binom{n}{k} 2^{n-k} = 3^n\). 
		The construction of each arc is done in \(O(n)\) time.
		A straight forward implementation thus finds the shortest path in \(3^{O(n)}\) time. 
		
		\medskip
		\noindent \emph{Correctness.} Observe that, by construction, the algorithm enumerates over all possible weakly triangular batches.
		By Lemmas~\ref{lem:continuous} to~\ref{lem:prefix}, we know that there exist an optimal solution, 
        such that every batch is weakly triangular.
        By Lemma~\ref{lem:batchGeom}, we know that every solution, such that every cycle consists of a weakly triangular batch is feasible.
        It immediately follows that Algorithm 1 computes an optimal solution.
	\end{proof}

    Implicitly enumerating the whole state space will incur an exponential computational cost in the number of targets. Thus, it may be helpful to discard states that will provably not lie on the shortest path between \((0, \ldots, 0)\) and \((1, \ldots, 1)\). We will make use of three dominance rules to reduce the state space, which will be explained in the following.
    
    \medskip \noindent \textbf{Dominance Rule 1:}
    Consider a state \(X_s\) with an unassigned, accessible target \(i\) for which it holds:
    \begin{enumerate}
        \item \(i\) is the only non-retrieved target in its stack;
        \item \(i\) is the highest non-retrieved target;
        \item \(i\) lies in the leftmost stack that still contains an non-retrieved target.
    \end{enumerate}
    Then retrieving \(i\) next is dominant: the only transition to consider from \(X_s\) is to the state \(X_t\) with \(x_i=1\).

    To see this consider that \(i\) can never be retrieved together with other targets in the same cycle, retrieving \(i\) cannot change accessibility for any other targets and retrieving it now will only reduce energy expenditure for any other target.
    
    \medskip 
    \noindent
    \textbf{Dominance Rule 2:}
    Consider a state \(X_s\) with two unassigned accessible targets \(i\) and \(i'\) for which it holds:
    \begin{enumerate}
        \item \(i\) is in the same stack as \(i'\);
        \item \(i\) is located directly on top of \(i'\).
    \end{enumerate}
    Then retrieving \(i'\) next is dominated: the transition from \(X_s\) to the state \(X_{t}\) with \(x_{i'}=1\) can be pruned.
    
    To see this consider that \(i'\) is retrieved at a higher cost than \(i\) (an additional box needs to be lifted) and that after the retrieval of \(i'\), \(i\) falls exactly into the place that was taken by \(i'\). Since targets are otherwise identical we could swap the label of target \(i\) and \(i'\) and reach the same state at lower cost by retrieving \(i\) first.

    \medskip 
    \noindent
	\textbf{Dominance Rule 3:}
    Consider a state \(X_s\) and two transitions to states \(X_t\) and \(X_t'\) with associated batches \(Y_t\) and \(Y_t’\), which differ only in one element (\(Y_t = Y_t' \bigcup \{i\}\)) for which it holds:
    \begin{enumerate}
        \item \(i\) is the only non-retrieved target in its stack;
        \item \(i\) is retrieved at the same height as another target \(i'\) in batch \(Y_t\);
        \item the stack of \(i\) is located to the right of the stack of \(i’\);
        \item there are no other non-retrieved targets in the stacks between those of \(i\) and \(i'\);
        \item retrieving \(i\) preserves accessibility of all remaining targets.
    \end{enumerate}
    Then the path from \(X_s\) to \(X_t'\) to \((1, \ldots, 1)\) cannot be shorter than the path between \(X_s\) to \(X_t\) to \((1, \ldots, 1)\).

    To see this consider that since all ULs until the stack of \(i'\) are already lifted to retrieve \(i'\), taking \(i\) in the same cycle only incurs additional energy of all the ULs in between the stack of \(i'\) and \(i\). Since there are no non-retrieved targets in these stacks it cannot incur less energy to retrieve \(i\) later instead of now. Since retrieving \(i\) further does not impact the accessibility of other boxes, it is always better to retrieve it now instead of later. Dominance Rule 3 would thus allow us to discard the transition from \(X_s\) to \(X_t'\). Moreover, if this transition from \(X_s\) to \(X_t'\) is on the shortest path from \((0, \ldots, 0)\) to \(X_t'\), then the whole state of \(X_t'\) can be discarded.

    We investigate the ability of the dominance rules to eliminated states from the enumeration in the experimental study.
	
    \subsection{Greedy Algorithm}\label{sec:greedy}
   Although the dynamic program and a state-of-the-art solver for Model 1 can compute optimal solutions, their runtimes may become prohibitive on large instances.
   We now propose a greedy algorithm that provides a good solution within negligible time.
   Furthermore, the algorithm is based on simple, iterative rules that makes it accessible for practitioners.

    Given the set \(X\) of already retrieved targets and the current batch \(Y\), we call a target \(b\) \emph{legal} if (i) \(b\) is accessible, (ii) \(b\) can be added to \(Y\) via one of Lemmas~\ref{lem:local1}--\ref{lem:local3}, and (iii) retrieving \(b\) preserves feasibility of the remaining instance by Lemma~\ref{lem:feasible}.
	
	The main idea of the algorithm is to sequentially retrieve the UL that is most critical in terms of accessibility, while retrieving other targets if they are legal.
	Specifically, let \(b\) denote the target at maximum height; in case of a tie, choose the one in the stack that is rightmost. 
    Initiate the algorithm with \(X \leftarrow \emptyset\) and \(c \leftarrow 0\).
	Start a new cycle by setting \(c \leftarrow c+1\). In stack \(s(b)\), identify the accessible targets and let \(b'\) be the topmost among them. Initialize the batch \(Y \leftarrow \{b'\}\) and iteratively add to \(Y\) any target that is legal. When no more legal targets exist, retrieve all targets in \(Y\) and let \(X \leftarrow X \cup Y\). If all targets are retrieved, terminate the algorithm, otherwise proceed to the next cycle.
	We refer to the procedure as Algorithm~2
	
	\begin{theorem}
	    Algorithm 2 computes a feasible solution in \(O(n^3)\) time.
	\end{theorem}
		
	\begin{proof}
	We discuss the time complexity and the correctness of the algorithm separately.
	\medskip 
	
	\noindent \emph{Time Complexity.} 
    Observe that the algorithm initiates at most \(n\) cycles. Selecting the anchor \(b\) in a cycle takes \(O(n)\) time. Identifying the next legal target \(b'\) requires \(O(n\log n)\) time and is performed at most \(n\) times overall, since each target is added at most once. Consequently, the total running time is \(O(n^2 \log n)\).
	
	\medskip
	\noindent \emph{Correctness.} At each retrieval, we apply Lemmas~\ref{lem:local1}–\ref{lem:local3} to maintain the batch’s weakly triangular structure and hence its retrievability by Lemma~\ref{lem:batchGeom}. 
    In addition, by verifying the feasibility condition in Lemma~\ref{lem:feasible}, we ensure that removing the chosen target preserves feasibility for the remaining instance.
    Hence, the algorithm maintains feasibility throughout and returns a feasible solution.
	\end{proof}	
    
	\section{Computational Study}\label{sec:compStudy}
	
	In this section, we computationally assess the tradeoffs between the runtime and solution quality of the MIP, DP, and greedy algorithm on instances ranging from small to large.
	We describe the testbed in Section~\ref{sec:testbed} and present results in Section~\ref{sec:comp_res}.

    All algorithms were implemented in C++ (Visual Studio 2017 version 15.9.8). The GUROBI API was used with GUROBI version 12.0.0. All tests were run on an AMD Ryzen Threadripper 3990X with 2.9 GHz (single-threaded) and 256 GB RAM. We imposed a time limit of 600 seconds for each algorithm to reflect practical computational constraints.
	
	\subsection{Instance Generation}\label{sec:testbed}
	As the side-access compact warehouses is a novel warehouse solution, our tests are based on synthetic data.
	Each instance is parameterized by a tuple \( (d, w, h) \), where \( d \) denotes the number of targets, \( w \) the maximum number of stacks, and \( h \) the maximum number of heights.
	The layout as well as the placement of the targets are generated randomly.
	  After generating an instance, we verify its feasibility and discard it if it is infeasible.
	
	We create two sets of instances, whose dimensions are selected to cover a range that reflects current industry standards.
	  The \emph{small instances} are defined by \(d\in\{5,10,15\}\), \(w\in\{8,12,16\}\), and \(h\in\{8,12,16\}\); the \emph{large instances} are defined by \(d\in\{18,21,24\}\), \(w\in\{20,24,28\}\), and \(h\in\{20,24,28\}\).
    We generate \(30\) feasible instances for each of the \(27\) configurations, yielding \(810\) feasible instances in total for each set.
	
	Tables~\ref{tab:small} and~\ref{tab:large} report the results for the small and large instances, respectively.
    The first four columns specify the instance parameters (number of targets, maximum width, maximum height) and the number of generated instances (“Count”).
    Columns 5–8 summarize the MIP results: the share of instances for which the best solution found by the MIP solver matches the DP optimum, the share of the proven optimals by the solver (zero optimality gap at termination), the average optimality gap at termination, and the average runtime in seconds.
    Columns 9–12 summarize the results of the DP: 
    the share of instances solved optimally within the time limit, the average total number of generated and explored states, and the average runtime.
    Finally, columns 13–15 show the greedy algorithm's performance: the share of instances where the objective matches the DP optimum, the average gap to the DP optimum, and the average runtime.
	
	\subsection{Computational Results}\label{sec:comp_res}
	
	We start by analyzing the performance of the algorithms on the small instances.
	The MIP solver was able to both find and prove optimal solutions in all instances with \( d \in \{5,10\} \). For \( d = 15 \), the solver still found optimal solutions (i.e. the found solution matched the solution found by the DP) but failed to prove them in several cases, resulting in an average proof rate of $91.1\%$ across all small instances. Despite the drop in proof rate, MIP consistently found the optimal solutions, with an average optimality gap of only $0.6\%$. Its runtimes scaled with instance size, averaging $95.7$ seconds and reaching several hundred seconds for instances with \( d = 15 \).
	
	The DP algorithm found optimal solutions for all small instances with negligible runtimes averaging only \( 0.015 \) seconds. 
    Moreover, the number of explored states remained small relative to the total state space, even for larger layout configurations. The fact that the total number of states grew exponentially while the number of explored states grew sublinearly shows the effectiveness of dominance rules and how they effectively prune the search space.
	
	In contrast, the greedy algorithm identified optimal solutions in only one fifth of small instances, with an average optimality gap of \( 20.8\% \). While it solved each instance in under \( 0.01 \) seconds, its solution quality was highly sensitive to instance structure.
	
	A comparison of the three methods shows clear performance differences in small instances. 
    The DP algorithm consistently achieved optimal solutions with negligible runtimes and relatively low state exploration. 
    The MIP solver, while significantly slower, found all optimal solutions but failed to prove optimality in \( 8.9\% \) of the instances within the 600-second time limit. 
    This suggests that the branch-and-bound search identifies high-quality solutions quickly, but closing the remaining optimality gap entails substantial certification cost.
    On the other hand, the greedy algorithm is able to compute a solution instantly with adequate performance.

	\begin{sidewaystable}
  \centering
  \scriptsize
  \resizebox{\linewidth}{!}{%
    \begin{tabular}{cccc|rrrr|rrrr|rrr}
      \toprule
      \multicolumn{4}{c}{} 
        & \multicolumn{4}{c}{MIP} 
        & \multicolumn{4}{c}{DP} 
        & \multicolumn{3}{c}{Heuristic} \\
      \cmidrule(lr){5-8}\cmidrule(lr){9-12}\cmidrule(lr){13-15}
        \#targets & Max.\ width & Max.\ height & Count
        & \shortstack{\#found\\opt. [\%]\textsuperscript{*}} & \shortstack{\#proven\\opt. [\%]} & \shortstack{Avg.\\gap [\%]} & \shortstack{Avg.\\runtime [s]}
        & \shortstack{\#found\\opt. [\%]} 
        & \shortstack{Avg.\# total\\states} & \shortstack{Avg.\# explored\\states} & \shortstack{Avg.\\runtime [s]}
        & \shortstack{\#found\\opt. [\%]} & \shortstack{Avg.\\gap [\%]} & \shortstack{Avg.\\runtime [s]} \\
      \midrule
  5 & 8 & 8 & 30 & 100.0 & 100.0 & 0.0 & 0.064 & 100.0 & 32 & 24 & 0.000 & 40.0 & 5.4 & 0.001 \\
  5 & 8 & 12 & 30 & 100.0 & 100.0 & 0.0 & 0.049 & 100.0 & 32 & 23 & 0.000 & 36.7 & 4.7 & 0.001 \\
  5 & 8 & 16 & 30 & 100.0 & 100.0 & 0.0 & 0.046 & 100.0 & 32 & 22 & 0.000 & 40.0 & 5.3 & 0.001 \\
  5 & 12 & 8 & 30 & 100.0 & 100.0 & 0.0 & 0.080 & 100.0 & 32 & 25 & 0.000 & 43.3 & 5.1 & 0.001 \\
  5 & 12 & 12 & 30 & 100.0 & 100.0 & 0.0 & 0.078 & 100.0 & 32 & 26 & 0.000 & 36.7 & 3.4 & 0.002 \\
  5 & 12 & 16 & 30 & 100.0 & 100.0 & 0.0 & 0.059 & 100.0 & 32 & 26 & 0.000 & 33.3 & 2.4 & 0.002 \\
  5 & 16 & 8 & 30 & 100.0 & 100.0 & 0.0 & 0.090 & 100.0 & 32 & 26 & 0.000 & 23.3 & 8.0 & 0.001 \\
  5 & 16 & 12 & 30 & 100.0 & 100.0 & 0.0 & 0.047 & 100.0 & 32 & 25 & 0.000 & 33.3 & 2.5 & 0.002 \\
  5 & 16 & 16 & 30 & 100.0 & 100.0 & 0.0 & 0.039 & 100.0 & 32 & 28 & 0.000 & 36.7 & 1.4 & 0.002 \\
  \addlinespace
  10 & 8 & 8 & 30 & 100.0 & 100.0 & 0.0 & 4.813 & 100.0 & 1024 & 665 & 0.000 & 0.0 & 37.1 & 0.002 \\
  10 & 8 & 12 & 30 & 100.0 & 100.0 & 0.0 & 2.385 & 100.0 & 1024 & 642 & 0.000 & 0.0 & 27.2 & 0.002 \\
  10 & 8 & 16 & 30 & 100.0 & 100.0 & 0.0 & 2.267 & 100.0 & 1024 & 573 & 0.000 & 6.7 & 20.9 & 0.002 \\
  10 & 12 & 8 & 30 & 100.0 & 100.0 & 0.0 & 5.448 & 100.0 & 1024 & 745 & 0.000 & 3.3 & 21.1 & 0.002 \\
  10 & 12 & 12 & 30 & 100.0 & 100.0 & 0.0 & 4.009 & 100.0 & 1024 & 696 & 0.000 & 0.0 & 13.5 & 0.002 \\
  10 & 12 & 16 & 30 & 100.0 & 100.0 & 0.0 & 2.523 & 100.0 & 1024 & 720 & 0.000 & 3.3 & 19.6 & 0.003 \\
  10 & 16 & 8 & 30 & 100.0 & 100.0 & 0.0 & 5.299 & 100.0 & 1024 & 784 & 0.000 & 0.0 & 19.4 & 0.002 \\
  10 & 16 & 12 & 30 & 100.0 & 100.0 & 0.0 & 4.233 & 100.0 & 1024 & 793 & 0.000 & 0.0 & 8.8 & 0.002 \\
  10 & 16 & 16 & 30 & 100.0 & 100.0 & 0.0 & 2.192 & 100.0 & 1024 & 712 & 0.004 & 3.3 & 7.5 & 0.003 \\
  \addlinespace
  15 & 8 & 8 & 30 & 100.0 & 70.0 & 5.3 & 346.974 & 100.0 & 32768 & 19939 & 0.003 & 0.0 & 71.8 & 0.003 \\
  15 & 8 & 12 & 30 & 100.0 & 86.7 & 0.7 & 193.630 & 100.0 & 32768 & 18369 & 0.004 & 0.0 & 52.1 & 0.003 \\
  15 & 8 & 16 & 30 & 100.0 & 96.7 & 0.1 & 121.213 & 100.0 & 32768 & 16513 & 0.004 & 0.0 & 35.9 & 0.003 \\
  15 & 12 & 8 & 30 & 100.0 & 50.0 & 4.0 & 447.693 & 100.0 & 32768 & 20636 & 0.003 & 0.0 & 38.1 & 0.003 \\
  15 & 12 & 12 & 30 & 100.0 & 86.7 & 0.4 & 184.939 & 100.0 & 32768 & 19190 & 0.003 & 0.0 & 40.7 & 0.003 \\
  15 & 12 & 16 & 30 & 100.0 & 83.3 & 0.3 & 183.064 & 100.0 & 32768 & 18743 & 0.003 & 0.0 & 32.1 & 0.003 \\
  15 & 16 & 8 & 30 & 100.0 & 50.0 & 2.9 & 454.148 & 100.0 & 32768 & 23411 & 0.003 & 0.0 & 31.2 & 0.003 \\
  15 & 16 & 12 & 30 & 100.0 & 70.0 & 1.1 & 321.706 & 100.0 & 32768 & 21425 & 0.003 & 0.0 & 29.1 & 0.003 \\
  15 & 16 & 16 & 30 & 100.0 & 66.7 & 0.7 & 297.937 & 100.0 & 32768 & 21728 & 0.364 & 3.3 & 18.3 & 0.004 \\
  \addlinespace
  All & -- & -- & 810 & 100.0 & 91.1 & 0.6 & 95.742 & 100.0 & 11275 & 6908 & 0.015 & 12.7 & 20.8 & 0.002 \\
      \bottomrule
    \end{tabular}%
  }
  \caption{Results for small instances.\\
  \scriptsize
  \textsuperscript{*}When DP times out and MIP fails to prove optimality, we cannot verify if MIP's solution is optimal. Such cases are counted as 'not found optimal'.}
  
  \label{tab:small}
\end{sidewaystable}

    The MIP solver matched the DP optima in nearly all of the large instances where DP successfully terminated, i.e. instances with \( d \in \{18,21\} \). Note that without DP solutions for instances with \( d \in \{24\} \), we could not verify MIP's optimality. Moreover,  MIP's ability to prove optimality declined sharply as problem size increased, with proof rates approaching zero for these largest instances.
    The average optimality gap increased from \( 0.6\% \) for small instances to \( 4.3\% \) for large instances. 
    The average runtime, however, rose sharply, with many instances hitting the time limit without proving optimality.
    
	Again, the DP algorithm found optimal solutions for all large instances with \( d \in \{18,21\} \) with an increased runtime of more than 100 seconds. 
    However, it failed to return a solution for any instance with \( d = 24 \) and hit the 600-second runtime limit for all instances in this subgroup. 
	
	The greedy algorithm maintained a near-constant runtime across all large instances. 
    It failed to identify any optimal solutions. Interestingly, the average relative gap (\( 16.2\% \)) improved compared to that of the small instances (\( 20.8\% \)).
    This is a curiosity and may be a statistical artifact of our small sample or it may point to a good performance of the greedy algorithm at certain densities. Investigating this phenomenon with larger sample sizes would be valuable future work.
	
	The results for the large instance set reinforce the trade-offs observed for the small instance set. 
    While MIP and DP both achieved optimality for almost all large instances except those with \( d = 24 \), their scalability is constrained by proof overhead and state-space explosion, respectively. 
    In particular, DP failed to return any solution for the largest instances, as the corresponding state space exceeded the limits of feasible exploration. 
    On the other hand, MIP consistently delivered feasible solutions with small gaps, even when it could not prove optimality, and could match all the solutions found by DP. 
    In contrast, the greedy algorithm is extremely fast and simple, producing adequate solutions.
    
    \begin{sidewaystable}
  \centering
  \scriptsize
  \resizebox{\linewidth}{!}{%
    \begin{tabular}{cccc|rrrr|rrrr|rrr}
      \toprule
      \multicolumn{4}{c}{} 
        & \multicolumn{4}{c}{MIP} 
        & \multicolumn{4}{c}{DP} 
        & \multicolumn{3}{c}{Heuristic} \\
      \cmidrule(lr){5-8}\cmidrule(lr){9-12}\cmidrule(lr){13-15}
        \#targets & Max.\ width & Max.\ height & Count
        & \shortstack{\#found\\opt. [\%]\textsuperscript{*}} & \shortstack{\#proven\\opt. [\%]} & \shortstack{Avg.\\gap [\%]} & \shortstack{Avg.\\runtime [s]}
        & \shortstack{\#found\\opt. [\%]} 
        & \shortstack{Avg.\# total\\states} & \shortstack{Avg.\# explored\\states} & \shortstack{Avg.\\runtime [s]}
        & \shortstack{\#found\\opt. [\%]} & \shortstack{Avg.\\gap [\%]} & \shortstack{Avg.\\runtime [s]} \\
      \midrule
  18 & 20 & 20 & 30 & 100.0 & 26.7 & 2.1 & 517.853 & 100.0 & 262144 & 173988 & 6.795 & 0.0 & 19.2 & 0.006 \\
  18 & 20 & 24 & 30 & 96.7 & 23.3 & 2.3 & 511.654 & 100.0 & 262144 & 169162 & 7.097 & 0.0 & 15.7 & 0.006 \\
  18 & 20 & 28 & 30 & 100.0 & 33.3 & 0.9 & 460.083 & 100.0 & 262144 & 177791 & 7.535 & 0.0 & 11.2 & 0.007 \\
  18 & 24 & 20 & 30 & 100.0 & 20.0 & 2.1 & 512.553 & 100.0 & 262144 & 176012 & 6.825 & 0.0 & 17.2 & 0.006 \\
  18 & 24 & 24 & 30 & 96.7 & 20.0 & 1.4 & 529.335 & 100.0 & 262144 & 195807 & 7.714 & 0.0 & 14.0 & 0.007 \\
  18 & 24 & 28 & 30 & 100.0 & 26.7 & 0.6 & 506.535 & 100.0 & 262144 & 200500 & 8.150 & 0.0 & 14.2 & 0.008 \\
  18 & 28 & 20 & 30 & 100.0 & 26.7 & 1.4 & 535.701 & 100.0 & 262144 & 196460 & 7.521 & 0.0 & 12.8 & 0.007 \\
  18 & 28 & 24 & 30 & 100.0 & 16.7 & 1.0 & 542.708 & 100.0 & 262144 & 173639 & 6.703 & 0.0 & 11.6 & 0.008 \\
  18 & 28 & 28 & 30 & 100.0 & 33.3 & 0.6 & 511.484 & 100.0 & 262144 & 199606 & 7.589 & 0.0 & 8.6 & 0.008 \\
  \addlinespace
  21 & 20 & 20 & 30 & 100.0 & 6.7 & 6.3 & 579.024 & 100.0 & 2097152 & 1485982 & 134.866 & 0.0 & 20.9 & 0.006 \\
  21 & 20 & 24 & 30 & 100.0 & 3.3 & 3.7 & 599.252 & 100.0 & 2097152 & 1350257 & 147.601 & 0.0 & 23.5 & 0.007 \\
  21 & 20 & 28 & 30 & 96.7 & 13.3 & 3.2 & 575.147 & 100.0 & 2097152 & 1440752 & 153.200 & 0.0 & 22.1 & 0.008 \\
  21 & 24 & 20 & 30 & 100.0 & 0.0 & 5.5 & 600.014 & 100.0 & 2097152 & 1558591 & 128.221 & 0.0 & 21.2 & 0.007 \\
  21 & 24 & 24 & 30 & 96.7 & 3.3 & 3.2 & 582.473 & 100.0 & 2097152 & 1405321 & 112.914 & 0.0 & 17.6 & 0.008 \\
  21 & 24 & 28 & 30 & 100.0 & 16.7 & 2.5 & 554.710 & 100.0 & 2097152 & 1457794 & 127.200 & 0.0 & 18.0 & 0.008 \\
  21 & 28 & 20 & 30 & 100.0 & 0.0 & 4.5 & 600.014 & 100.0 & 2097152 & 1517776 & 123.560 & 0.0 & 15.7 & 0.007 \\
  21 & 28 & 24 & 30 & 100.0 & 0.0 & 3.4 & 599.951 & 100.0 & 2097152 & 1610235 & 148.077 & 0.0 & 14.7 & 0.008 \\
  21 & 28 & 28 & 30 & 100.0 & 3.3 & 2.3 & 591.612 & 100.0 & 2097152 & 1415768 & 114.047 & 0.0 & 13.3 & 0.009 \\
  \addlinespace
  24 & 20 & 20 & 30 & - & 0.0 & 12.8 & 599.955 & 0.0 & 16777216 & 661996 & 600.345 & - & - & 0.007 \\
  24 & 20 & 24 & 30 & - & 0.0 & 8.8 & 599.955 & 0.0 & 16777216 & 726202 & 600.440 & - & - & 0.008 \\
  24 & 20 & 28 & 30 & - & 3.3 & 5.3 & 587.432 & 0.0 & 16777216 & 657947 & 600.269 & - & - & 0.009 \\
  24 & 24 & 20 & 30 & - & 0.0 & 10.2 & 599.955 & 0.0 & 16777216 & 809895 & 600.438 & - & - & 0.008 \\
  24 & 24 & 24 & 30 & - & 0.0 & 7.0 & 600.014 & 0.0 & 16777216 & 819782 & 600.402 & - & - & 0.008 \\
  24 & 24 & 28 & 30 & - & 0.0 & 4.9 & 599.957 & 0.0 & 16777216 & 865538 & 600.279 & - & - & 0.005 \\
  24 & 28 & 20 & 30 & - & 0.0 & 9.9 & 599.958 & 0.0 & 16777216 & 1135960 & 600.434 & - & - & 0.002 \\
  24 & 28 & 24 & 30 & - & 0.0 & 6.1 & 600.015 & 0.0 & 16777216 & 1101188 & 600.502 & - & - & 0.002 \\
  24 & 28 & 28 & 30 & - & 0.0 & 3.8 & 599.956 & 0.0 & 16777216 & 1041597 & 600.289 & - & - & 0.002 \\
  \addlinespace
  All & -- & -- & 810 & 66.23 & 10.2 & 4.3 & 566.6 & 66.7 & 6378837 & 841687 & 247 & 0.0 & 16.2 & 0.007 \\
      \bottomrule
    \end{tabular}%
  }
  \caption{Results for large instances.\\
  \scriptsize
  \textsuperscript{*}When DP times out and MIP fails to prove optimality, we cannot verify if MIP's solution is optimal. Such cases are counted as 'not found optimal'.}
  
  \label{tab:large}
\end{sidewaystable}
	
	\section{Conclusion and Outlook}
	\label{sec:conclusions}
	
	In this paper, we initiated the algorithmic study of compact storage systems with side access, and introduced and analyzed the Side-Access Compact Retrieval Problem (SACRP). 
    We showed that the problem is strongly NP-complete and derived a number of structural insights. 
    Further, we showed that these can be exploited by algorithms. 
    The computational study shows that the DP consistently achieves optimal solutions with negligible runtimes for small to medium instances, while the MIP, though slower, reliably finds optimal solutions even when it cannot prove optimality. 
    On the other hand, the greedy algorithm provides adequate solution quality with near-instant computation times.

    Our results demonstrate that problem-specific structural insights can make exact algorithms viable even for problem classes where they are traditionally considered impractical. Despite exact methods typically reaching computational limits quickly on combinatorial problems of this scale, both our DP and MIP approaches successfully delivered optimal solutions for the majority of instances within a modest 10-minute time limit. This success stems from exploiting dominance rules and problem geometry, suggesting that with sufficient understanding of a problem's inherent structure, exact solution methods can remain competitive alternatives to heuristics. The quality of solutions obtained in short periods reinforces that compact storage systems with side access are not only theoretically interesting but also computationally tractable for practical deployment.
    
    We close the paper with some open problems and future research directions. First, our study focused on a restricted layout with lifts on a single side. A natural extension is the most general case with lifts on all sides. We conjecture that our geometric structural results can be generalized to this setting; however, the computational implications for adapting the exact algorithms remain unclear and require further research. Second, to assess the practical effectiveness of side-access compact storage systems, a direct benchmark against AutoStore, which is the currently most widely deployed compact storage system, is essential. 
    Third, and in the same vein, evaluating both systems under other objectives, such as time efficiency and robot utilization, would provide a fuller picture. Lastly, idle periods could be used to proactively rebalance the grid to a target layout under a limited lift budget in practice. Hence, developing structural characterizations of “good” layouts and algorithms to achieve them under limited energy/time budgets is of practical value and may be of theoretical interest.

	\section{Acknowledgement}
	This research is funded by dtec.bw – Digitalization and Technology Research Center of the Bundeswehr. dtec.bw is funded by the European Union – NextGenerationEU.

	\bibliographystyle{abbrv}
	\bibliography{lit.bib}

\end{document}